\documentclass[11pt]{article}

\usepackage{amssymb}
\usepackage{amsmath}
\usepackage{amsthm}
\usepackage{fancyhdr}
\usepackage{graphics}
\usepackage{graphicx}
\usepackage{environ}
\usepackage{framed}
\usepackage{url}
\usepackage{authblk}
\usepackage[margin=1.0in]{geometry}

\newcommand{\F}{\mathbb{F}}
\newcommand{\Pt}{\mathrm{P}}

\newcommand{\AM}{\mathrm{AM}}
\newcommand{\MA}{\mathrm{MA}}
\newcommand{\IP}{\mathrm{IP}}
\newcommand{\NP}{\mathrm{NP}}

\newcommand{\psd}{\mathrm{psd}}

\theoremstyle{plain}
\newtheorem{lemma}{Lemma}[section]
\newtheorem*{lemma*}{Lemma}
\newtheorem{theorem}[lemma]{Theorem}
\newtheorem{corollary}[lemma]{Corollary}
\newtheorem*{corollary*}{Corollary}

\theoremstyle{definition}
\newtheorem{definition}[lemma]{Definition}

\title{Fine Grained Pseudo-Deterministic Proofs}
\author[1]{Michel Goemans}
\author[2]{Shafi Goldwasser}
\author[3]{Dhiraj Holden}
\affil[1]{Department of Mathematics, MIT goemans@math.mit.edu}
\affil[2]{Simons Institute for the Theory of Computing, UC Berkeley shafi.goldwasser@gmail.com}
\affil[3]{Department of Electrical Engineering and Computer Science, MIT dholden@mit.edu}

\begin{document}
\maketitle
\abstract{

In this paper we study {\it doubly-efficient pseudo-deterministic proofs} for polynomial time search problems:  interactive proofs where a polynomial time prover can convince a probabilistic verifier that a solution to a search problem is a "canonical" solution and where the verifier verification time is "little-o" of the complexity of {\it finding} any solution to the search problem,
canonical or otherwise. This extends the notion of pseudo-deterministic proofs of \cite{GGH18} for non-polynomial time problems which did not require the prover to run in polynomial time.  

We exhibit doubly-efficient pseudo-deterministic algorithms for a host of natural and well studied problems, in which the verifier run time to verify that a solution is canonical is significantly faster than the best known algorithm for finding a solution as follows:

\begin{itemize}

\item

We show a doubly-efficient pseudo-deterministic $\NP$ proof for {\bf  linear programming} where the canonical solution which the prover will provide is the lexicographically greatest optimal solution for the LP. To this end, we show how through strong duality and perturbing the linear program this canonical solution can be both computed efficiently by the prover, and verified by the verifier.   The time of the verifier is $O(d^2 )$ for a linear program with integer data and at most $d$ variables and constraints, whereas the time to solve such linear program is $\tilde{O}(d^{\omega} )$ by any known randomized algorithms \cite{Cohen18} for $\omega$  the exponent for fast matrix multiplication .

\item
In the context of fine-grained search problems, we show:
\begin{enumerate}
    \item 

A doubly-efficient pseudo-deterministic $\NP$ proof for {\bf 3-SUM }and problems reducible to 3-SUM where the prover is a $O(n^2)$ time algorithm and the verifier takes time $\tilde{O}(n^{1.5})$. 

    \item

A doubly-efficient pseudo-deterministic $\NP$ proof for the {\bf  hitting set problem} where the verifier runs in time $\tilde{O}(m)$ and the prover runs in time $\tilde{O}(m^2)$ where $ m = \sum_{S \in \mathcal{S}} |S| + \sum_{T \in \mathcal{T}} |T|$ for inputs  collections of sets $\mathcal{S}, \mathcal{T}$.
    \item
A doubly-efficient pseudo-deterministic $\NP$ proof for the {\bf Zero Weight Triangle problem} where the verifier runs in time $\tilde{O}(n^{2 + \omega/3})$ and the prover runs in randomized time $\tilde{O}(n^3)$. The Zero Weight Triangle problem is equivalent to the {\bf All-Pairs Shortest Path problem}, a well-studied problem that is the foundation of many hardness results in graph algorithms \cite{VW13,WW10}, under sub-cubic reductions. 

\item 
A doubly-efficient pseudo-deterministic $\MA$ proof for the {\bf Orthogonal Vectors problem}, which is the problem of finding a pair of vectors that are orthogonal given $n$ vectors of length $d$ ,  where the verifier runs in time $\tilde{O}(nd)$ and the prover runs in time $\tilde{O}(n^2d^2)$. In addition, we show a doubly-efficient pseudo-deterministic $\MA$  proof for $k$-Clique where the verifier runs in time $\tilde{O}n^{\lfloor k/2 \rfloor + 2}$ and the prover runs in time $\tilde{O}(n^k)$. 
\end{enumerate}
\end{itemize}

}
\newpage
\begin{section}{Introduction}
Pseudo-deterministic algorithms, introduced by Gat and Goldwasser \cite{GG11},  
are probabilistic (polynomial-time) algorithms for search problems that, with high probability, {\it find} a unique output for each input except with negligible error probability.  Such output for input $x$ is referred to as the ''canonical'' output for  $x$.
Algorithms that satisfy the aforementioned condition are of importance whenever uniqueness or "reproducibility" of the answer is important. This is of
particular relevance in a distributed or parallel setting when an algorithm is executed by multiple parties for whom it is challenging (for reasons of trust or efficiency requirement) to {\it agree} on a common sequence of unbiased random coins.  

More recently, Goldwasser, Grossman and Holden \cite{GGH18} extended the study to pseudo-deterministic
interactive proofs for search problems, denoted $\psd\IP$. 
The new goal was to
{\it prove}  to a probabilistic polynomial time verifier that a solution to a search problem is canonical.
The motivation was to address those search problems for which polynomial time algorithms are not known and for which many solutions are possible, such as for graph isomorphism. In this case the
 {\it search problem} is to find an isomorphism between two graphs if one exists and an example of a  {\it canonical solution} would be the lexicographically smallest isomorphism.
One may think of the powerful prover as aiding the probabilistic polynomial time verifier to find canonical solutions to search problems, 
with high probability over the randomness of the verifier.  The challenge is that a malicious prover should not be able to convince the verifier to accept any solution other than the unique canonical one using a constant round interaction.\footnote{ If unbounded number of rounds are allowed, the $IP=PSPACE$ characterization implies that $\psd\IP = IP$.}

In this paper, we turn our attention to studying pseudo-deterministic proofs with the extra requirement that the prover runs in polynomial time in the complexity of the problem and the verifier can {\it verify} that a solution provided by the prover is canonical significantly more efficiently than solving the problem without the presence of the prover [ in particular, in all the protocols we design, the verifier runs  in "little-o" time of the best known algorithms for finding a solution). We call these {\it doubly-efficient pseudo-deterministic proofs}. Our aim  is to use doubly-efficient pseudo-deterministic proofs for {\it polynomial time problems}, where the prover runs in polynomial time in the complexity of the problem. 
Indeed,  for all of the doubly-efficient pseudo-deterministic proofs presented below, with the exception of linear programming, the runtime of the prover is at most a constant times the runtime of the best known deterministic algorithm.

\subsection{Our Results}

\subsubsection*{A new notion: Doubly-efficient pseudo-deterministic interactive proofs}

We define  {\it doubly-efficient pseudo-deterministic interactive proofs} for a search problem $R$ of time complexity $T(n)$ (consisting of pairs
$(instance,solution)$)  with associated canonization function $c$ as 
a pair of interacting algorithms:   a probabilistic polynomial time prover  which runs in time $poly(T(n))$ and 
a probabilistic verifier which runs in time $o(T(n))$ which on a common input instance $x$ engage in constant number of rounds of interaction at the end of which with high probability 
the verifier outputs a canonical solution $y=c(x)$ if any solution exists and otherwise rejects $x$.
Analogously to the case of completeness in interactive proofs for languages, we require that for every input $x$, there exists an honest prover which can send the correct solution $c(x)$ to the verifier when one exists. Analogously to the case of soundness, no dishonest prover can cause the verifier to output a solution other than $c(x)$ (the canonical one) (except with very low probability).

A few remarks are in order. 
\begin{itemize}
\item
Naturally this question is particularly interesting for search problems for which a lower bound on its worst case complexity $T(n)$  is known or has been widely conjectured. 
This will drive our choice of problems for which we show doubly efficient pseudo-deterministic proofs.

\item The setting of doubly-efficient interactive proofs naturally models a cryptographic settings
where users wish to have access
to common cryptographic system-wide keys or parameters, such as a  pair $(g,p)$ for ${Z_p}$ with prime $p$  and generator $g$ for a given input length $n$.
A  central authority (with additional computational power) can of course choose the common system-wide parameter and broadcast it to all, but then who is to say that the central party did not chose its randomness in a way that would force an output for which the trusted center knew some ``trapdoor'' information which would enable it to break the underlying cryptographic security?  Viewing the generation of a cryptographic key as a solution to a search problem $R$ per security parameter, a doubly-efficient pseudo-deterministic proof for $R$ would ensure that  the prover had no choice in which parameter to broadcast as he could prove that his solution is canonical.

\item Doubly-efficient pseudo-deterministic proofs for search problems $R$ with associated
canonization function $c$ are closely related to {\bf computation delegation} of computing $c(x)$ on input $x$. 
The delegation problem was posed by Goldwasser, Kalai, and Rothblum \cite{GKR08} and  become known under the name doubly-efficient interactive proof systems. 
The difference in the requirements is that \cite{GKR08}
require the verifier to run in linear (up to log factors) time and addresses deterministic computations. 
Doubly-efficient interactive proofs have been shown by \cite{GKR08}  for
log-space uniform sets in NC (or, more generally, to inputs that are acceptable by log-space uniform bounded-depth circuits, where the number of rounds in the proof system is linearly related to the depth of the circuit). Reingold, Rothblum and Rothblum \cite{RRR}
showed that any set  decidable in polynomial-time by an algorithm of space complexity $s(n)\leq n^{0.499}$, has a constant-round interactive proof system in which the prover runs  in polynomial time and the verifier runs in time $\tilde{O}(n)$.
Finally Goldreich and Rothblum \cite{GR18}  show direct constructions of doubly-efficient interactive proof systems for problems in $\Pt$ that are believed to have 
relatively high complexity such as  $t$-CLIQUE and $t$-SUM. 
In our work we turn our attention to a host of natural problem whose upper bound has been under very wide investigation for decades.

We remark that works  on proof systems and delegation did not stay within the realm of theory
alone. Rather, they became the theoretical basis for several system implementations of a delegation system as
they offered reasonably efficiently realizable protocols. Indeed, there is a flourishing literature surrounding
the refinement and implementation of these theoretical protocols \cite{BackesBFR15,BackesFR13,Ben-SassonCG0MTV14,Ben-SassonCGTV13,Ben-SassonCTV14,BraunFRSBW13,ChiesaTV15,CormodeMT12,
CostelloFHKKNPZ15,FioreGP14,KosbaPPSST14,ParnoHG016,SettyBVBPW12,SettyMBW12,SettyVPBBW12,Thaler13,ThalerRMP12,VuSBW13,WahbySRBW15} (see \cite{WalfishB15} for a survey). 
Similarly, We hope that the presently proposed study of doubly-efficient pseudo-deterministic proofs will impact practice (and beyond).

\item Our results about fine-grained complexity problems build on the existence of doubly-efficient proofs for the existence and non-existence of solutions, shown in a paper by Williams and a paper by Carmosino, Gao, Impagliazzo, Mihajlin, Paturi, and Schneider \cite{CGI16}, \cite{W16}. In particular, finding doubly-efficient proofs of nonexistence for certain problems would give strong evidence for the existence of doubly-efficient pseudo-deterministic proofs in places where we do not believe there are faster deterministic algorithms. 
\end{itemize}

\subsubsection*{Doubly-efficient pseudo-deterministic algorithms for linear programming and fine-grained complexity problems}

\paragraph{Linear Programming:} We show a doubly-efficient pseudo-deterministic proof for the linear programming problem.  Verifying an optimal solution to a linear programming problem can be done thanks to {\it strong duality}: there exists a solution to the dual problem with the same value as the solution to the primal problem. We show that a special optimal solution, namely the lexicographically greatest solution, can be efficiently obtained by the prover, and that the prover can convince the
verifier  that the LP solution it gives to the verifier is indeed the lexicographically greatest solution; this is done through perturbing the linear program and strong duality. 
More concretely, every linear program (say, where the objective is to maximize) has a corresponding dual linear program, a minimization problem, with the property that (i) (weak duality) any feasible solution to the dual provides an upper bound on the optimal primal value and (ii) (strong duality) there exists an optimal solution to the dual with the same value as the primal optimal solution.  Furthermore, there exist compact polynomial-sized solutions to the primal and dual linear programs. Therefore such a  polynomial-sized feasible solution to the dual with an equal value as a primal solution provides a compact certificate for the optimality of this primal solution. 

The currently best known time to solve a  linear program with integer data and at most $d$ variables and constraints is $\tilde{O}(d^{\omega})$ randomized \cite{Cohen18} where  $\omega$  corresponds to
the exponent for fast matrix multiplication which is  currently at $\approx 2.37$ and $\tilde{O}()$ hides polylog factors including a $\log(1/\delta)$ factor to account for the accuracy $\delta$ in solving the linear program. The time of the verifier to
verify a pair of primal and dual optimal solution is only $O(d^2)$
 as this only requires matrix-vector multiplication.

\paragraph{Central Problems studied in fine-grained complexity:}
We  show doubly-efficient pseudo-deterministic proofs for several fine-grained complexity problems where the verifier significantly beats the conjectured time. The challenge is to find a proof where the prover's running time is (almost) the same as the running time of the deterministic algorithm to find any solution (within a polylogarithmic factor of the simplest deterministic algorithm), not necessarily canonical. In the case of two of our problems, making the running time close to the running time of the deterministic algorithm requires the prover to run in a randomized fashion. We also remark that in all cases the canonical solution is the lexicographically smallest solution.
 
The {\bf 3-SUM problem} has an easy $O(n^2)$ time algorithm which can be improved by polylogarithmic factors. It is an outstanding open question whether there is an algorithm that significantly improves  $O(n^2)$. Finding such an algorithm would yield algorithms for a host of other problems in computational geometry \cite{GO95, BGO97}.   Here, we show a doubly-efficient pseudo-deterministic proof that outputs the lexicographically first such triple of elements where the verifier takes time $\tilde{O}(n^{1.5})$ and the prover runs in time $\tilde{O}(n^{2})$.  We crucially use the fact that  \cite{CGI16} gives a nondeterministic proof that there is no triple of elements that sum to 0 where the verifier takes time $\tilde{O}(n^{1.5})$. 

The {\bf hitting set problem} is the problem of finding a set in a collection of sets that intersects every set in a different collection of sets. We show 
 a pseudo-deterministic proof for the hitting set problem where the verifier runs in time $\tilde{O}(m)$ and the prover runs in time $\tilde{O}(m^2)$ where
 $ m = \sum_{S \in \mathcal{S}} |S| + \sum_{T \in \mathcal{T}} |T|$ for inputs $\mathcal{S}, \mathcal{T}$ collections of sets.
 This problem has been conjectured to take $m^{2 - o(1)}$ time \cite{V15}. 

The {\bf All-Pairs Shortest Path } problem is a well-studied problem that is the foundation of many hardness results in graph algorithms \cite{VW13,WW10}. In particular, the {\bf Zero Weight Triangle} problem is equivalent to the All-Pairs Shortest Path problem under subcubic reductions. 
We show a doubly efficient pseudo-deterministic proof for the Zero Weight Triangle problem  where the verifier runs in time  $\tilde{O}({n^{2 + \omega/3}})$ and the prover runs in randomized time $\tilde{O}(n^3)$. It is believe that there does not exist a deterministic solution to this problem running in time $n^{3 - \epsilon}$, suggesting that the running time of the prover is close to the optimal running time for a deterministic algorithm.

The {\bf Orthogonal Vectors} problem was one of the first problems studied in fine-grained complexity, and is the basis for almost all hardness results based on the hardness of SETH (Strong Exponential Time Hypothesis). We show a doubly-efficient pseudo-deterministic $\MA$ proof for the Orthogonal Vectors problem where the verifier runs in time $\tilde{O}(nd)$ and the prover runs in time $\tilde{O}(n^2d^2)$. It is believed that there does not exist a deterministic solution to this problem running in time $n^{2 - \epsilon}2^{o(d)}$, suggesting that the running time of the prover is close to the optimal running time for a deterministic algorithm.

The {\bf k-Clique} problem is the problem of finding a clique of size $k$ in a graph. We show a doubly-efficient pseudo-deterministic $\MA$ proof for the $k$-Clique problem where the verifier runs in time $\tilde{O}(n^{\lfloor k/2 \rfloor + 2})$ and the prover runs in time $\tilde{O}(n^k)$.

\subsubsection*{Outline of the Techniques}

Our techniques take on the following flavor: For a search problem $R$, defined by $R(x,y)$ if and only if $y$ is a solution to $x$, the pseudo-deterministic algorithm, given $x$, finds the lexicographically first $y$ such that  $R(x,y)$. To do this, it asks whether there exists $y'$ such that $(x,0y') \in R$, $y'$ such that $(x,1y') \in R$, etc. and finds the first $y$ such that $R(x,y)$ recursively. 
The notion of "lexicographically first" generalizes to
allow other orderings and other encodings of the input.
This suggests that more generally doubly-efficient pseudo-deterministic proofs for search are
the ones where there is a doubly-efficient proof of {\bf existence} and
a doubly-efficient proof of {\bf nonexistence} of solutions to said search problem. A more general theorem (Lemma \ref{lemma}) follows under general conditions.

\end{section}

\begin{section}{Preliminaries}
In this section we will introduce concepts needed to give pseudo-deterministic proofs that improve on the best known deterministic algorithms for problems studied in fine-grained complexity. 

\begin{definition} [Search Problem] 
A {\it search problem} is a relation $R$ consisting of pairs $(x,y)$ and we define $L_{R}$ to be the set of $x$ such that $\exists y (x,y) \in R$. 
\end{definition}

The goal of an algorithm solving a search problem is to find a $y$ such that $(x,y) \in R$.   The focus of pseudo-determinism is to give algorithms for search problems that find canonical solutions; a pseudo-deterministic algorithm will output the same solution to a search problem with high probability over its randomness. \cite{GGH18} extended the notion of pseudo-determinism to interactive proofs and brought the concept of $\NP$ search problems with unique answers under the umbrella of pseudo-determinism. We will refer to this work's definition of a pseudo-deterministic proof. The pseudo-deterministic proofs in our setting will always either output the unique solution or $\bot$. 

\begin{definition}[Pseudo-deterministic proof \cite{GGH18}]
A search problem $R$ is in \textit{pseudo-deterministic} $\IP$ (often denoted $\psd\IP$) if there exists a function $c$ where all $x \in L_R$ satisfy $(x, c(x)) \in R$, and an interactive protocol between a probabilistic polynomial time verifier algorithm $V$ and
a prover (unbounded algorithm) $P$ such that for every $x \in L_R$:
\begin{enumerate}
\item (Canonical Completeness)
There exists a $P$ such that $\Pr_r[(P,V)(x,r)=c(x)] \ge {2\over 3}$. (We use $(P, V)(x, r)$ to denote the output of the verifier $V$ when interacting with prover $P$ on input $x$ using randomness $r$).
\item (Canonical Soundness)
For all $P'$, $\Pr_r[(P',V)(x,r)=c(x)$ or $\bot]\ge{2\over 3}$.
\end{enumerate}
And (Standard Soundness) for every $x \notin L_R$, for all provers $P'$, $\Pr_r[(P',V)(x,r) \neq \bot] \le {1\over 3}$.
\end{definition}

This is analogous to the definition of pseudo-deterministic NP, except we allow the prover and verifier to interact. In the setting we consider, the prover and verifier both run in polynomial time, with the prover given more time than the verifier. Our goal is to construct pseudo-deterministic proofs for problems such that the verifier runs in time faster than the best known deterministic algorithm for the problem. 

\end{section}

\begin{section}{Doubly-efficient pseudo-deterministic proofs}
We want to extend the concept of pseudo-deterministic proofs to the setting where the prover also runs in polynomial time, and we want to extend the concept of doubly-efficient interactive proofs to the setting where the verifier outputs a unique solution. Both of these tasks are accomplished by introducing {\it doubly-efficient pseudo-deterministic proofs} : proofs where both the verifier and prover run in polynomial time, the verifier running in time asymptotically faster, and where the verifier will output a unique solution given an input. 

\begin{definition}
A $(t_1(n),t_2(n))$ pseudo-deterministic proof is a pseudo-deterministic proof where the verifier $V$ runs in (probabilistic) time $t_1(n)$ and the prover $P$ runs in (probabilistic) time $t_2(n)$.
\end{definition}

Ideally, we want the prover to run in time almost equal to the deterministic running time of the problem, as this means the total work is not much more than the work of solving this problem deterministically. However, we say that  pseudo-deterministic proof is {\it non-trivial} as long as the verifier runs faster than the deterministic running time of the problem.
To demonstrate the concept, we will consider the pseudo-deterministic proof for graph isomorphism. The prover from \cite{GGH18} only needs the power to solve graph isomorphism. We know from \cite{B16} that graph isomorphism is in quasi-polynomial time. Thus, the result of \cite{GGH18} about graph isomorphism can be restated as:

\begin{corollary}
Graph Isomorphism has a $(poly(n),quasipoly(n))$ pseudo-deterministic proof. 
\end{corollary}

A large class of pseudo-deterministic algorithms have the following format: for a search problem $R$, the pseudo-deterministic algorithm, given $x$, finds the lexicographically first $y$ such that $R(x,y)$. To do this, it asks whether there exists $y'$ such that $(x,0y') \in R$, $y'$ such that $(x,1y') \in R$, etc. and finds the first $y$ such that $R(x,y)$ recursively. For instance, \cite{GG11} gives a pseudo-deterministic algorithm for testing if a polynomial is non-zero by finding the lexicographically first non-zero solution. Given $p(x_1,...,x_n)$, the algorithm tests if $p(0,...,x_n)$ is zero everywhere. If it is not zero everywhere, then the algorithm checks if $p(0,0,...,x_n)$ is zero everywhere, and otherwise the algorithm checks if $p(1,...,x_n)$ is zero everywhere. This continues recursively until the algorithm finds the first element that is non-zero or rejects. Also, \cite{GGH18} provides a pseudo-deterministic proof for graph isomorphism where the verifier outputs the lexicographically first isomorphism by going recursively. The algorithm starts by figuring out where the first vertex is mapped in the lexicographically first isomorphism by looping through the vertices, then where the second vertex is mapped, and so on until the lexicographically first isomorphism has been found.  

We will use a structure inspired by this  to define doubly-efficient pseudo-deterministic proofs for a large class of problems studied within the fine grained complexity literature.

\end{section}

\begin{section}{Linear programming}

In \cite{GGH18}, we use the fact that graph non-isomorphism has an $\AM$ proof to give a pseudo-deterministic $\AM$ proof for graph isomorphism. Here we show a pseudo-deterministic proof for linear programming. Linear programming is the class of optimization problems with linear constraints and a linear objective function. 
We exploit the fact that linear programming admits a good characterization, a compact way of certifying the optimality of a solution. Indeed every linear program (say, where the objective is to maximize) has a corresponding dual linear program, a minimization problem, with the property that (i) (weak duality) any feasible solution to the dual provides an upper bound on the optimal primal value and (ii) (strong duality) there exists an optimal solution to the dual with the same value as the primal optimal solution.  Furthermore, there exist compact polynomial-sized solutions to the primal and dual linear programs. Therefore such a  polynomial-sized feasible solution to the dual with an equal value as a primal solution provides a compact certificate for the optimality of this primal solution. 

In order to be able to turn this into a pseudo-deterministic proof, we need the prover to identify a special, unique optimal solution (as there could be a continuum of primal optimal solutions), and provide a way for the verifier to efficiently verify it. As special solution, we use the {\it lexicographically greatest} optimal solution to the primal. Among all optimal solutions, the lexicographically greatest first maximizes $x_1$, then $x_2$, and so on; see below for a precise definition. To verify it, one option would be to provide dual optimal solutions to a squence of dual linear programs corresponding to the definition of lexicographically greatest maximal solution. A better (more efficient) way, which we describe in this section, is to show that we can perturb the objective function of the primal linear program in such a way that there is a unique optimal solution and that this solution is the unique lexicographically greatest optimal solution for the unperturbed linear program. 

We start with basic notation and linear programming fundamentals. 

\begin{definition}
A linear program is the problem $\max \{\mathbf{c}^{\top}\mathbf{x}\}$ subject to the constraints $A\mathbf{x} \leq \mathbf{b}$ and $\mathbf{x} \geq \mathbf{0}$. Its dual is the linear program $\min \{\mathbf{b}^{\top}\mathbf{y}\}$ subject to the constraints $A^\top \mathbf{y} \geq \mathbf{c}$ and $\mathbf{y} \geq \mathbf{0}$.
\end{definition}

\begin{theorem}
(Weak duality) If $\mathbf{x},\mathbf{y}$ are feasible solutions to a linear program given by $\max \{\mathbf{c}^{\top}\mathbf{x}\}$ subject to $A\mathbf{x} \leq \mathbf{b}$ and $\mathbf{x} \geq \mathbf{0}$ and its dual respectively, then $\mathbf{c}^{\top}\mathbf{x} \leq \mathbf{b}^{\top}\mathbf{y}$. (Strong duality) Furthermore $\mathbf{x},\mathbf{y}$ are optimal solutions if and only if $\mathbf{c}^{\top}\mathbf{x} = \mathbf{b}^{\top}\mathbf{y}$. 
\end{theorem}

Furthermore, there exist optimal solutions of polynomial size, since any extreme point (which cannot be expressed as a strict convex combination of feasible points) has this property. 

\begin{theorem}[\cite{GoemansNotes}] \label{thsize}
Let $P$ be the linear program given by $\max \mathbf{c}^{\top}\mathbf{x}$ subject to $A\mathbf{x} \leq \mathbf{b}$, where all inputs are integers and $A$ is an $m \times n$ matrix. Define $L = m + n + \log (\max_{A'} |det(A')|) + \log(\max_{i} |b_i|) + \log(\max_{j} |c_j|)$, where $A'$ range over all square submatrices of $A$. Then any extreme point $\mathbf{x}$ of $P$ is of the form $x_i=\frac{p_i}{q}$ where $q$ and $p_i$'s are integers satisfying $1\leq q <2^L$ and $0\leq p_i<2^L$ for all $i$.
\end{theorem}
This quantity $L$ is often used when referring to efficiency of linear programming algorithms, and can be seen (see \cite{GoemansNotes}) to be polynomially related to the binary encoding of all the input data.  

\begin{definition}
The {\it lexicographically greatest} optimal solution $\mathbf{x^*}$ to a linear program  $\max \{\mathbf{c}^{\top}\mathbf{x}\}$ subject to $A\mathbf{x} \leq \mathbf{b}$ and $\mathbf{x} \geq \mathbf{0}$ is the solution that satisfies (i) feasibility: $A\mathbf{x^*} \leq \mathbf{b}$ and $\mathbf{x^*} \geq \mathbf{0}$, (ii) optimality: $\mathbf{c}^{\top}\mathbf{x^*} = \max_{A\mathbf{x} \leq \mathbf{b},\mathbf{x} \geq \mathbf{0}} \{\mathbf{c}^{\top}\mathbf{x}\} $, and (iii) for every $\mathbf{x}\in \arg\max_{A\mathbf{x} \leq \mathbf{b},\mathbf{x} \geq \mathbf{0}} \{\mathbf{c}^{\top}\mathbf{x}\}$, either $\mathbf{x} = \mathbf{x^*}$ or there exists $i\leq n$ with $x_i<x^*_i$ and $x_j=x^*_j$ for $j<i$. 
\end{definition}

Now that we have defined the necessary terminology, we can proceed to proving that linear programming has a pseudo-deterministic interactive proof. To do so, we perturb our linear program so that the only optimal solution to the new linear program is the lexicographically greatest solution to the original program, and then use the dual linear program to prove that the solution given to the verifier is optimal. 

\begin{theorem}
\label{LP}
Let $P$ be the linear program given by $\max \mathbf{c}^{\top}\mathbf{x}$ subject to $A\mathbf{x} \leq \mathbf{b}$, with $L$ defined as above. Then, the linear program $P'$ given by $\max \mathbf{c}^{\top}\mathbf{x} + \epsilon x_1 + \epsilon^2 x_2 + ... + \epsilon^n x_n$, where $\epsilon = 2^{-3L-2}$, has a unique solution which is the lexicographically greatest solution of $P$. 
\end{theorem}

\begin{proof}
First consider the unperturbed linear program $P$, and two extreme point solutions $\mathbf{x}^{(1)}$ and $\mathbf{x}^{(2)}$, with corresponding denominators $q_1$ and $q_2$ respectively (see Theorem 
\ref{thsize}).

If $\mathbf{c}^{\top}\mathbf{x}^{(1)}>\mathbf{c}^{\top}\mathbf{x}^{(2)}$ then $\mathbf{c}^{\top}\mathbf{x}^{(1)}-\mathbf{c}^{\top}\mathbf{x}^{(2)} \geq \frac{1}{q_1q_2} >2^{-2L}$. Let $\mathbf{c'}$ be the perturbed $\mathbf{c}$ (by adding the vector $(\epsilon, \epsilon^2,\cdots,\epsilon^n)$). Then 
$$\mathbf{c'}^{\top}\mathbf{x}^{(1)}-\mathbf{c'}^{\top}\mathbf{x}^{(2)} > 2^{-2L} +\sum_{i=1}^n \epsilon^i (x^{(1)}_i-x^{(2)}_i) > 2^{-2L} - 2^L \sum_{i=1}^n \epsilon^i > 2^{-2L} - 2^L \epsilon/(1-\epsilon)>0, 
$$ given our choice of $\epsilon$. This shows that, after perturbation, we still have that $\mathbf{x}^{(1)}$ has a greater objective value than $\mathbf{x}^{(2)}$. 

Suppose, on the other hand, that $\mathbf{c}^{\top}\mathbf{x}^{(1)}=\mathbf{c}^{\top}\mathbf{x}^{(2)}$ and that $\mathbf{x}^{(1)}$ is lexicographically greater than $\mathbf{x}^{(2)}$, i.e. that $x_i^1>x_i^2$ while $x_j^1=x_j^2$ for $j<i$. Then 
$$\mathbf{c'}^{\top}\mathbf{x}^{(1)}-\mathbf{c'}^{\top}\mathbf{x}^{(2)} = \sum_{k=i}^n \epsilon^k (x_k^1-x_k^2) \geq \epsilon^i \left(\frac{1}{q_1q_2} -\sum_{\ell=1}^{n-i} \epsilon^{\ell} 2^L\right) >\epsilon^i \left(2^{-2L} -\frac{\epsilon}{1-\epsilon} 2^L\right)>0,
$$ showing that, after perturbation, the lexicographically greater solution $\mathbf{x}^{(1)}$ has greater (perturbed) objective function value. Together, this shows that the unique optimal solution to the perturbed problem is the lexicographically greatest solution to $P$.
\end{proof}

Observe that the parameter $L'$ of the perturbed linear program increases polynomially to $O(nL)$, but the precision needed to solve the linear program approximately in order to be able to recover the unique extreme point solution is still $2^{-O(L)}$ , as this represents a lower bound on the difference in value between any two extreme point solutions. 

\begin{theorem}
There exists a $(O(n^2 \log(1/\delta)),\tilde{O}(n^{\omega} \log(1/\delta)))$ pseudo-deterministic interactive proof for finding an optimal solution to a linear program $P$. 
\end{theorem}

Linear programs with at most $n$ variables and constraints can be solved within an error of $\delta$ in time $\tilde{O}(n^{2.5}\log(1/\delta))$ deterministically \cite{Vaidya89}, and in time $\tilde{O}(n^{\omega}\log(1/\delta))$ randomized \cite{Cohen18} with $\omega$ (currently $\sim 2.37$) corresponds to the exponent for fast matrix multiplication.  The notation $\tilde{O}$ hides polylog factors. The time to verify a pair of primal and dual optimal solution is only $O(d^2)$ (with a $\log(1/\delta)$ factor for bit complexity) as this only requires matrix vector multiplication. So, verification is currently more efficient than finding the solution.

\begin{proof}
By Theorem \ref{LP}, we can perturb the objective function of $P$ and obtain a linear program $P'$ which has a unique solution, namely the lexicographically greatest solution of $P$. Let $Q'$ be the dual linear program to $P'$. The prover sends over optimal solutions to $P'$ and $Q'$. Then, the verifier checks to see whether the solutions are feasible and also whether the value of the solution to $P'$ is equal to the value of the solution to $Q'$. If both of these conditions hold, the verifier outputs the solution to $P'$, otherwise it outputs $\bot$. If the prover is honest, then clearly the verifier will output the solution to $P'$. A cheating prover cannot make the verifier output a different solution to $P$, as this would not correspond to an optimal solution of $P'$ since it is unique. 
\end{proof}

\end{section}

\begin{section}{Problems studied in fine-grained complexity}

\begin{subsection}{Orthogonal Vectors and $k$-Clique}

The Orthogonal Vectors problem is the problem of finding a pair of orthogonal vectors, given a list of $n$ vectors of length $d$. The Orthogonal Vectors conjecture is that this problem cannot be solved in time $O(n^{2 - \epsilon} 2^{o(d)})$, and is the basis for a large number of results in fine-grained complexity. In \cite{W16}, Williams gives an MA proof for counting the number of orthogonal vectors where the verifier takes time $\tilde{O}(nd)$. We will show that the prover in this proof can run in time $\tilde{O}(n^2d^2)$. 

\begin{theorem}[Theorem A.1 of \cite{W16}]
Let $d \leq n$. Then, for every $V \subseteq \{0,1\}^d$ such that $|V| = n$, there is an $\MA$-proof system certifying for every $v \in V$ the number of vectors $u \in V$ such that $\langle u,v \rangle = 0$, with verifier running time $\tilde{O}(nd)$ and error $1/poly(n)$.
\end{theorem}

We will give a sketch of the $\MA$ proof of \cite{W16}, where we will set the error probability to $\epsilon = 1/poly(n)$. We want to find two vectors in $V$ that are orthogonal. The proof relies on a circuit $C$ that takes a vector in $\{0,1\}^d$ and computes the number of vectors in $V$ orthogonal to it. This circuit has size $O(nd)$ and degree at most $2d$. The goal of the interactive proof is to evaluate $C$ over $\F_p$, where $p$ is a prime greater than $n^2 \cdot d$, on all of the vectors in $V$. This gives the number of pairs of vectors orthogonal to each other. In order to evaluate $C$ on many points, the proof transforms $C$ into a univariate polynomial over $\F_{p^l}$, where $l$ is the smallest integer such that $p^l > n(2d)/\epsilon$. To do so, we associate the vectors $v^1,...,v^n$ with elements of $\F_{p^l}$ $a_1,...,a_n$ in a natural way, and then consider the degree $n$ polynomial $\psi_i$ with $\psi_i(a_j) = v_i^j$, the $i$th coordinate of the $j$th vector, which can be found via polynomial interpolation. The univariate polynomial that we evaluate is then $C(\psi_1(x),\psi_2(x),...,\psi_d(x))$. The circuit for this polynomial can be constructed in $\tilde{O}(nd)$ time. The prover sends the coefficients of this polynomial and then the verifier checks that these coefficients are correct by evaluating the polynomial on a random point. (The prover also has to send an irreducible polynomial of degree $l$ over $\F_p$ to the verifier, but we will not go into that here.) Since the degree of the polynomial is at most $n(2d)$, the probability that this check fails is at most $\epsilon$. Then, the verifier can compute the value of the polynomial on $a_1,...,a_n$ using the coefficients in time $\tilde{O}(nd)$ by fast multi-point evaluation. 

\begin{theorem}
Orthogonal Vectors has a $(\tilde{O}(nd), \tilde{O}(n^2d^2))$ pseudo-deterministic $\MA$ proof. 
\end{theorem}

\begin{proof}
We will give an $\MA$ proof that a pair of orthogonal vectors is the lexicographically first pair. The pseudo-determinisic $\MA$ proof will consist of the following steps: \begin{enumerate}
    \item The prover sends the pair of orthogonal vectors $(u,v)$ that the prover claims is the lexicographically first pair and also the coefficients of the polynomial described above. 
    \item The verifier checks that the pair of vectors $(u,v)$. is orthogonal. If the pair is not orthogonal, the verifier outputs $\bot$.  
    \item The verifier computes the number of vectors orthogonal to each vector as in the proof above. \item The verifier makes sure that there are no vectors before $u$ orthogonal to any of the other vectors by looking at the value of the polynomial at those vectors, and also that there are no vectors before $v$ that $u$ is orthogonal to by checking directly. If both of these statements hold, the verifier outputs $u,v$ otherwise it outputs $\bot$. 
    \end{enumerate} 
    This gives an $\MA$ proof that the pair of orthogonal vectors given is the lexicographically first pair. To see that the verifier runs in time $\tilde{O}(nd)$, observe that checking that a pair of vectors is orthogonal takes time $\tilde{O}(d)$, computing the polynomial as above takes time $\tilde{O}(nd)$, and checking that there are no vectors before $u$ orthogonal to any of the other vectors and that there are no vectors before $v$ that $u$ is orthogonal to takes time $\tilde{O}(nd)$, as the first check is done by looking at the values of the polynomial on vectors before $u$ and the second check is done by directly computing the inner product of $u$ with all vectors before $v$. It remains to check that the prover can run in time $\tilde{O}(n^2d^2)$. The prover has to find the lexicographically first pair of orthogonal vectors, which can be done in time $\tilde{O}(n^2d)$ using brute force, and compute the coefficients of $C(\psi_1(x),\psi_2(x),...,\psi_d(x))$, a polynomial of degree at most $nd$. This circuit has size $\tilde{O}(nd)$, so computing the coefficients using interpolation takes time $\tilde{O}(n^2d^2)$. 
\end{proof}

The $k$-Clique problem is, given a graph $G$, to find a clique of size $k$ in $G$. \cite{W16} gives an $\MA$ proof to certify the number of $k$-cliques where the verifier takes time $\tilde{O}(n^{\lfloor k/2 \rfloor + 2})$ using the same technique as for Orthogonal Vectors. In this case the polynomial $P$ used is the sum over all $\lfloor k/2 \rfloor$-cliques of the $\lceil k/2 \rceil$th elementary symmetric polynomial of the joint neighborhood (i.e. the set of vertices that are adjacent to every vertex in the $\lfloor k/2 \rfloor$-clique) of the $\lfloor k/2 \rfloor$-clique. This polynomial has degree $O(n)$ and size $O(n^2 \binom{n}{\lfloor k/2 \rfloor})$, and it needs to be evaluated on $\binom{n}{\lceil k/2 \rceil}$ points to compute the number of $k$-cliques. By a similar argument to the Orthogonal Vectors polynomial, the prover in the $\MA$ proof can be made to run in time $\tilde{O}(n^k)$. 

\begin{theorem}
$k$-Clique has a $(\tilde{O}(n^{\lfloor k/2 \rfloor + 2}k), \tilde{O}(n^k))$ pseudo-deterministic $\MA$ proof.
\end{theorem}
\begin{proof}
We outline the pseudo-deterministic $\MA$ proof below. Let $G = (V,E)$ be the graph for which we want to find a $k$-clique. The interactive proof proceeds as follows: 

\begin{enumerate}
    \item The prover sends the lexicographically first $k$-clique $(v_1,...,v_k)$ and a proof that it is the lexicographically first one to the verifier. The proof is the polynomial $P$ above but modified so that it only sums over  $\lfloor k/2 \rfloor$-cliques that contain $v_1,...,v_{i-1}$ and all vertices before $v_i$, and a polynomial is send for every $1 \leq i \leq \lfloor k/2 \rfloor$. 
    \item The verifier checks whether $(v_1,...,v_k)$ is a $k$-clique and outputs $\bot$ if it is not a $k$-clique.
    \item The verifier checks whether the proof is correct and outputs $(v_1,...,v_k)$ if it is correct and $\bot$ otherwise.
\end{enumerate} 
Without loss of generality let $(v_1,...,v_k)$ be the lexicographically first $k$-clique with $v_1 < v_2 ... < v_k$. To verify that the $k$-clique $(v_1,...,v_k)$ is the lexicographically first one, we need to check that $(v_1,...,v_k)$ is a clique and there is no $k$-clique with any vertex $u$ before $v_1$, there is no $k$-clique with $v_1$ and any vertex $u$ before $v_2$, and so on. It is easy for the verifier to check whether $(v_1,...,v_k)$ is a $k$-clique. To prove there is no $k$-clique before it, we will have to modify the polynomial $P$ from above. For $i$ from 1 to $\lfloor k/2 \rfloor$, we will modify the polynomial to only sum over $\lfloor k/2 \rfloor$-cliques that contain $v_1,...,v_{i-1}$ and all vertices before $v_i$ instead of all $\lfloor k/2 \rfloor$-cliques as this sum will tell us whether there are $k$-cliques containing $v_1,...,v_{i-1}$ and a vertices before $v_i$, which would be a clique lexicographically before $(v_1,...,v_k)$. We showed above that there exists an $\MA$-proof certifying the value of this polynomial at the necessary points where the prover takes time $\tilde{O}(n^k)$ and the $\tilde{O}(n^{\lfloor k/2 \rfloor + 2})$. For $i$ from $\lfloor k/2 \rfloor + 1$ to $k$, the verifier can check whether there are $k$-cliques containing $v_1,...,v_{i-1}$ and all vertices before $v_i$ in time $\tilde{O}(n^{\lfloor k/2 \rfloor + 2})$ using brute force. If all of the checks hold, the verifier outputs $(v_1,...,v_k)$ otherwise it outputs $\bot$. This gives a pseudo-deterministic $\MA$ proof of $k$-Clique.
\end{proof}

\end{subsection}

\begin{subsection}{3-SUM and problems reducible to 3-SUM}

3-SUM is the problem to find 3 numbers that sum to 0, where the numbers are drawn from 3 lists. The 3-SUM problem has an easy $O(n^2)$ time algorithm and this can be improved by polylogarithmic factors \cite{chan2018more}. It is an outstanding open question whether there is an algorithm that is much faster than $O(n^2)$, and finding such an algorithm would give faster algorithms for a host of other problems in computational geometry \cite{GO95, BGO97}. We will show a pseudo-deterministic proof where the verifier runs in time $\tilde{O}(n^{1.5})$. 

\begin{definition}
We say the 3-SUM problem is the problem of, given 3 lists $a_1,...,a_n$, $b_1,...,b_n$, $c_1,...,c_n$, of $O(\log n)$ bit integers, finding a triple $a_i,b_j,c_k$ such that $a_i + b_j + c_k = 0$. 
\end{definition}

In addition, \cite{CGI16} gives a nondeterministic proof that there is no triple of elements that sum to 0 where the verifier takes time $\tilde{O}(n^{1.5})$. 

\begin{theorem}[\cite{CGI16}]
3-SUM $\in \mathrm{coNTIME}(\tilde{O}(n^{1.5}))$.
\end{theorem}

\begin{theorem}
\label{3SUM}
There exists a non-deterministic proof for $\overline{\text{3-SUM}}$ where the verifier runs in time $\tilde{O}(n^{1.5})$ and the prover runs in expected time $\tilde{O}(n^2)$.
\end{theorem}

\begin{proof}
There is a proof of this fact when the prover is allowed an unbounded amount of time in \cite{CGI16}. Here we modify the proof in such a way that the prover runs in expected time $\tilde{O}(n^2)$. The prover will do the following to find a proof that there is no triple of elements that sums to 0.

On input $(a_1,...,a_n),(b_1,...,b_n),(c_1,...,c_n)$, 
\begin{enumerate}
    \item The prover computes the first $n^{1.5}$ primes. This takes time $\tilde{O}(n^{1.5})$ as only $\tilde{O}(n^{1.5})$ numbers need to be checked to find the first $n^{1.5}$ primes. 
    \item The prover selects a random prime $p$ in the first $n^{1.5}$ primes and checks whether there are less than $n^{1.5}\log^2(n)$ triples $(a_i,b_j,c_k)$ such that $a_i + b_j + c_k = 0 \pmod{p}$. If there are more, the prover tries primes until they find a $p$ such that there are fewer triples that sum to $0\pmod p$.
    \item The prover sends the verifier $p$, the number of triples $(a_i,b_j,c_k)$ such that $a_i + b_j + c_k = 0 \pmod{p}$ $t$, and the set of triples $(a_i,b_j,c_k)$ such that $a_i + b_j + c_k = 0 \pmod{p}$.
\end{enumerate}
 Since there are $n^3$ triples and each sum must be a product of at most $\log(n)$ primes, we get that there are at most $n^3\log n$ pairs of triples and primes $((a_i,b_j,c_k),p)$ such that $a_i + b_j + c_k = 0 \pmod p$. Thus, in the first $n^{1.5}$ primes, over half the primes $p$ will have less than $n^{1.5} \log n$ triples that sum to 0 mod $p$ by Markov's inequality. Thus if we sample a random prime in the first $n^{1.5}$ primes, we will get a good prime with high probability. Computing the sums and checking whether they are equal to 0 mod $p$ still takes time $\tilde{O}(n^2)$ deterministically.
 
 The verifier behaves in the same way as in \cite{CGI16}:
 
 On input $(a_1,...,a_n),(b_1,...,b_n),(c_1,...,c_n)$,
 \begin{enumerate}
     \item The verifier receives $p$, $t$, and $\{(i,j,k)\}$ from the prover. Then, the verifier rejects unless $p$ is prime and the number of triples $(i,j,k)$ is $t$.
     \item The verifier uses Fast Fourier Transform (as in \cite{CGI16}) to compute the number of triples $a_i,b_j,c_k$ such that $a_i + b_j + c_k = 0 \pmod p$ and rejects unless there are $t$ triples.
     \item The verifier computes $a_{i_l} + b_{j_l} + c_{k_l}$ and $a_{i_l} + b_{j_l} + c_{k_l} \pmod p$ for every $(i_l,j_l,k_l)$ in the set of triples $\{(i,j,k)\}$ and rejects unless every triple has the property that $a_{i_l}+ b_{j_l} + c_{k_l} = 0 \pmod p$ and $a_{i_l}+ b_{j_l} + c_{k_l} \neq 0$.
 \end{enumerate}
 
 The soundness of this proof follows from \cite{CGI16}.  
\end{proof}

With this, we can construct a pseudo-deterministic proof for 3-SUM where the prover runs in time almost equal to the best known deterministic algorithm for 3-SUM. 

\begin{theorem}
3-SUM has a $(\tilde{O}(n^{1.5}),\tilde{O}(n^2))$ pseudo-deterministic proof. 
\end{theorem}

\begin{proof}
To give a pseudo-deterministic proof the prover must prove that a given solution is a solution and it is a canonical solution. For 3-SUM, the canonical solution will be the lexicographically first solution; this is the triple $(a_i,b_j,c_k)$ with $a_i + b_j + c_k = 0$ such that for any other triple $(a_{i'},b_{j'}, c_{k'})$ with $a_{i'} + b_{j'} + c_{k'} = 0$ either $i' > i$, $i' = i$ and $j' > j$, or $i' = i$, $j' = j$, and $k' > k$. The interactive proof will proceed as follows. 
\begin{enumerate} \item The prover sends $a_i,b_j,c_k$, a proof that the lists $(a_1,...,a_{i-1}),( b_1,...,b_n),(c_1,...,c_n)$ do not have a 3-SUM, a proof that the lists $(a_i),( b_1,...,b_{j-1}),(c_1,...,c_n)$ do not have a 3-SUM, and a proof that the lists $(a_i),( b_j),(c_1,...,c_{k-1})$ do not have a 3-SUM. 
\item The verifier then checks that $a_i + b_j + c_k = 0$ and that all of the proofs are correct and outputs $(a_i,b_j,c_k)$ if both of these conditions hold otherwise it outputs $\bot$.
\end{enumerate} 
By Theorem \ref{3SUM} this can be done with the prover running in time $\tilde{O}(n^2)$ and the verifier running in time $\tilde{O}(n^{1.5})$.

\end{proof}

\begin{corollary}
Determining whether there are three collinear points in a set of points on the plane has a $(\tilde{O}(n^{1.5}),\tilde{O}(n^2))$ pseudo-deterministic proof.  
\end{corollary}

\end{subsection}

\begin{subsection}{Hitting Set}

The Hitting Set problem is, given two collections of sets, find a set in the first collection that intersects every set in the second collection. The Hitting Set problem is also conjectured to take $m^{2 - o(1)}$ time \cite{V15}. Here we give a pseudo-deterministic proof in which the verifier runs in linear time.

\begin{definition}
The Hitting Set problem is, given two collections $\mathcal{S}, \mathcal{T}$ of sets, find a set $S$ such that $S \cap T \neq \emptyset \forall T \in \mathcal{T}$.
\end{definition}
\begin{theorem}[\cite{CGI16}] \label{hittingset}
There is a nondeterministic proof where the verifier runs in time $O(m), m = \sum_{S \in \mathcal{S}} |S| + \sum_{T \in \mathcal{T}} |T|$, and the prover runs in time $O(m^2)$ for the Hitting Set problem and the complement of the Hitting Set problem. 
 \end{theorem} 
\begin{proof}
 The nondeterministic proof for Hitting Set proceeds as follows. The prover uses brute force to find a set $S$ such that $S \cap T \neq \emptyset \forall T \in \mathcal{T}$ in time $O(m^2)$ and sends it to the verifier, along with a list of $u_T$ such that $u_T \in S \cap T$. The verifier can check whether $S \cap T \neq \emptyset \forall T \in \mathcal{T}$ in time $O(m)$ by looping over every $T$ and checking that $u_T \in S \cap T$. 
 
 The nondeterministic proof for the complement of Hitting Set proceeds as follows. The prover sends, for every $S \in \mathcal{S}$, a set $T_{S} \in \mathcal{T}$ such that $S \cap T = \emptyset$. The verifier checks that for every $S$ $S \cap T_{S}$ is empty. The prover can find such a set $T$ for every $S$ in time $O(m^2)$ and the verifier runs in time $O(m)$. 
\end{proof}
 
\begin{theorem}
Hitting Set has a $(O(m),O(m^2))$ pseudo-deterministic proof.
\end{theorem} 
 
\begin{proof}
The canonical solution will be the first set $S$ such that $S \in \mathcal{S}$ is a hitting set for $\mathcal{T}$. The interactive proof goes as follows:
\begin{enumerate}
\item The prover sends a proof that the first hitting set $S$ is a hitting set.
\item The prover also sends a proof that the collection $\mathcal{S}'$ of all sets before $S$ in $\mathcal{S}$ does not have a hitting set for $\mathcal{T}$, as in Theorem \ref{hittingset}.
\item The verifier checks that $S$ is a hitting set and outputs $\bot $ if it is not a hitting set. 
\item The verifier checks that the proof that the collection of all sets before $S$ does not have a hitting set for $T$ is correct. If $S$ is a hitting set and the proof is correct, the verifier outputs $S$ otherwise it outputs $\bot$.
\end{enumerate}
Since the proofs for Hitting Set and its complement have the verifier run in time $O(m)$ and the prover run in time $O(m^2)$, so does this proof. To see that this is a pseudo-deterministic proof, the verifier will only accept if $S$ is a hitting set and if the proof that there are no hitting sets before $S$ is correct. 
\end{proof}

\end{subsection}

\begin{subsection}{Model checking of graph properties}

The problem of model checking is to determine whether a graph has a property expressed as a first-order formula over edge predicates. A large number of different graph problems can be expressed as model checking of first-order properties as observed by \cite{CGI16}. For instance both the $k$-Dominating Set problem \cite{PW10} and asking whether a graph has diameter 2 \cite{BCH16} can be written as model checking problems. \cite{W14} shows that given a first-order property of a graph with $k$ quantifiers over vertices, checking whether the graph has this property can be done in time $\tilde{O}(n^{k - 3 + \omega})$. We extend the work of \cite{CGI16} on sparse graphs to provide pseudo-deterministic proofs. 

\begin{definition}
We say a graph property is a formula $Q_1 x_1 \in X_1 Q_2 x_2 \in X_2...Q_k x_k \in X_k \psi$, where $\psi$ is a quantifier-free formula on edge predicates and the model checking problem for a graph property is to determine whether the property holds for a given graph.
\end{definition}

\begin{theorem}[\cite{CGI16}]
If a formula with $k$ quantifiers does not have the form $\exists^{k-1}\forall$, then the model checking problem for the formula can be solved in co-nondeterministic time $m^{k-2}$ where $m$ is the number of edges in the graph. 
\end{theorem}

\begin{theorem}[\cite{CGI16}]
The deterministic complexity of model checking a $k$-quantifier formula is $O(m^{k-1})$. 
\end{theorem}

\begin{theorem}
If a formula does not have the form $\exists^{k-1}\forall$, there exists a $(O(m^{k-2}),O(m^{k-1}))$ pseudo-deterministic proof for finding a setting to the first set of existential quantifiers of that formula.
\end{theorem}

\begin{proof}
Suppose the formula has $i$ existential quantifiers at the beginning. The interactive proof proceeds as follows. \begin{enumerate} \item The prover sends $x_1,...,x_i$ that is the lexicographically first set such that $Q_{i+1} x_{i+1}...Q_kx_k \psi(x_1,...,x_i)$ is true. 
\item The prover also sends a proof that for any $1 \leq j \leq i$, $\exists x'_j < x_j Q_{j+1} x_{j+1}...Q_kx_k \psi(x_1,...,x_{j-1})$ is false.

\item The verifier checks whether $Q_{i+1} x_{i+1}...Q_kx_k \psi(x_1,...,x_i)$  using brute force and outputs $\bot$ if $Q_{i+1} x_{i+1}...Q_kx_k \psi(x_1,...,x_i)$ is false.
\item The verifier checks the proofs that for any $1 \leq j \leq i$, $\exists x'_j < x_j Q_{j+1} x_{j+1}...Q_kx_k \psi(x_1,...,x_{j-1})$ is false and accepts if the proofs are correct, otherwise it outputs $\bot$. 

\end{enumerate}

If the first $i$ quantifiers are $\exists$, then we can find $x_1,...,x_i$ such that $Q_{i+1} x_{i+1}...Q_kx_k \psi(x_1,...,x_i)$ nondeterministically in time $O(m^{k-i})$ for any $1 \leq j \leq i$, and we can check for any $1 \leq j \leq i$ that $\exists x'_j < x_j Q_{j+1} x_{j+1}...Q_kx_k \psi(x_1,...,x_{j-1})$ in co-nondeterministic time $O(m^{k - 2})$ by setting $X'_{j} = X_{j} \cap \{ x | x < x_j \}$. For both of these checks, the prover has to solve a model checking problem with at most $k$ quantifiers, which has complexity $O(m^{k-1})$. This shows that there is a $O(m^{k-2})$ pseudo-deterministic proof where the prover runs in time $O(m^{k-1})$ for finding a setting to the first set of existential quantifiers of a formula, if the formula does not have the form $\exists^{k-1}\forall$.

\end{proof}

\end{subsection}

\begin{subsection}{Problems equivalent to All-Pairs Shortest Path}

The All-Pairs Shortest Path problem has been the focus of much research in fine-grained complexity. It has been shown by \cite{VW13,WW10} that many problems related to graphs reduce to the All-Pairs Shortest Path problem and vice versa, so finding a faster algorithm for any one of these problems would yield a fast algorithm for a host of graph problems. \cite{CGI16} shows that the Zero Weight Triangle problem, which is equivalent to the All-Pairs Shortest Path problem under subcubic reductions \cite{WW10}, has a $O(n^{3 - \epsilon})$ co-nondeterministic algorithm, which is faster than all known deterministic algorithms. We use this to construct a pseudo-deterministic proof for the Zero Weight Triangle problem.  

\begin{definition}
The Zero Weight Triangle problem is given a graph $G = (V,E)$ and edge weights $e(i,j)$, find $i,j,k \in V$ such that $e(i,j) + e(i,k) + e(j,k) = 0$. 
\end{definition}

\begin{theorem}[\cite{CGI16}]
The Zero Weight Triangle problem has a nondeterministic proof and a co-nondeterministic proof where the verifier runs in time $O(n^{2 + \omega/3})$, where $\omega$ is the largest number such that matrix multiplication is in time $O(n^{\omega})$. 
\end{theorem}

\begin{theorem}
The Zero Weight Triangle problem has an $(\tilde{O}(n^{2 + \omega/3}),\tilde{O}(n^3))$ pseudo-deterministic proof.
\end{theorem}

\begin{proof}
 There is an easy reduction from Zero Weight Triangle to Zero Weight Triangle on tripartite graphs. Then, we remove all edges in the first column going from $i' \geq i$ to $j$, and thus the resulting graph has a triangle with zero weight iff there exists a triangle in the original graph with zero weight and $i' < i$, where $i$ is the smallest vertex in the claimed lexicographically first zero weight triangle. A similar argument as the argument showing the prover for 3-SUM runs in randomized time $\tilde{O}(n^2)$ shows that the prover for the pseudo-deterministic proof of Zero Weight Triangle runs in randomized time $\tilde{O}(n^3)$. Specifically, both proofs in \cite{CGI16} use a random prime $p$ and Fast Fourier Transform to quickly compute sums $\pmod p$. The verifier for the proof in $\cite{CGI16}$ takes time The interactive proof proceeds as follows.
 \begin{enumerate}
     \item The prover sends a triangle $i,j,k$ that is the claimed lexicographically first triangle of zero weight, as well as a proof that there is no $(i',j',k')$ with zero weight such that $i' < i$.
     \item The verifier checks whether $i,j,k$ is a zero weight triangle and checks the proof that there is no $(i',j',k')$ with zero weight such that $i' < i$. 
     \item The verifier uses brute force to make sure there is no $j',k'$ such that either $j' < j$ or $j' = j$ and $k' < k$ and $(i,j',k')$ is a zero weight triangle. If all of these checks pass the verifier outputs $(i,j,k)$ otherwise it outputs $\bot$. 
 \end{enumerate}
 To see that the verifier runs in time $\tilde{O}(n^{2 + \omega/3})$, that is the running time for the proof in \cite{CGI16} that a graph does not have a triangle with zero weight. 
\end{proof}
\end{subsection}

\end{section}

\section{General Conditions}

\begin{lemma}
\label{lemma}
Suppose we have a search problem $R(x,y)$ such that $|y| = poly(x)$, finding the lexicographically first $y$ given $x$ such that $R(x,y)$ takes time $t_1(n)$, computing $R(x,y)$ takes time $t_2(n)$, and $y$ can be written as $y_1...y_k$ such that the following holds: 
\begin{itemize}
\item Given $x,y_1,...,y_i$, the problem $\exists z_i < y_i \exists y_{i+1},...,y_k R(x,y_1,...,y_{i-1},z_i,y_{i+1}...y_k)$ can be solved in co-nondeterministic time $t_3(n)$ where the prover runs in time $t_4(n)$.
\end{itemize}
Then there exists a $(t_2(n) + k * t_3(n), t_1(n) + k * t_4(n))$ pseudo-deterministic proof that outputs the lexicographically first $y$ such that $R(x,y)$.
\end{lemma}
\begin{proof}
Our algorithm proceeds in two stages: in the first stage, the prover gives $y$, taking time $t_1(n)$ to find the lexicographically first $y$ such that $R(x,y)$, and the verifier checks whether $R(x,y)$ and outputs $\bot$ otherwise; this takes time $t_2(n)$. In the next stage we prove that $y$ is the lexicographically first such $y$; that is, for all $z <_{lex} y$, $\neg R(x,z)$. To do so, we only need to check that there is no $i$ such that $\exists z_i < y_i \exists z_{i+1},...,z_k \text{ such that } R(y_1,y_2,...,y_{i-1},z_i,z_{i+1},...,z_k)$ for $1 \leq i \leq k$. Since we have to do this $k$ times, the total time of this stage is $k * (t_2(n))$ for the verifier and $k * t_4(n)$ for the prover. Our algorithm clearly outputs the lexicographically first $y$ such that $R(x,y)$, and a cheating prover cannot make the verifier output a different $y$. 
\end{proof}

\begin{section}{Conclusions and Open Problems}
We defined the notion of doubly-efficient pseudo-deterministic proofs and gave a number of examples of search problems for which we showed doubly-efficient pseudo-deterministic proofs. In all of these cases,  the verifier runs faster than the best known probabilistic algorithm for the problem which can offer significant improvements for settings in which a more powerful computer (cloud, special purpose device, centralized authority) can perform the computation first and prove it to a significantly less powerful user.
In all these cases the prover's computation increases by a factor subpolynomial in $n$ from what is necessary to find any solution of the problem at hand, even without requiring that the solution be canonical ( e.g. without requiring the solution to be lexicographically smallest).
An interesting problem would be show that this is true in general. Namely, that for any doubly-efficient pseudo-deterministic proof  the computation of the prover need be no more than whats necessary to find the canonical solution. 
Finally,  we remark that in all the cases we treated,  the canonical solution was  the lexicographically smallest (or largest as in the LP case) but other canonical solutions are possible.
\end{section}

\section*{Acknowledgments}
We would like to thank Andrea Lincoln for preliminary discussion on pseudo-deterministic proofs for fine-grained complexity. Dhiraj Holden was supported by NSF MACS - CNS-1413920 and Michel Goemans was partially supported under ONR Contract N00017-17-1-2177.

\bibliographystyle{plain}
\bibliography{RRR}

\end{document}